\documentclass[aps,prl,twocolumn,showpacs,superscriptaddress]{revtex4-2}

\pdfoutput=1
\usepackage[utf8]{inputenc}
\usepackage{
graphicx,
color,
amsmath,
amsfonts,
enumerate,
amsthm,
amssymb,
mathtools,
enumitem,
thmtools,
mathrsfs,
hyperref,
subfigure,
mathdots,
enumitem,
centernot,
bm,
soul,
bbm}
\usepackage[normalem]{ulem}
\usepackage{booktabs}
\usepackage{url}
\usepackage{float}
\usepackage{tkz-euclide}
\usepackage{csquotes}
\usepackage{bbold}
\usepackage{pgfplots}
\usepackage[capitalise, noabbrev]{cleveref}
\usepackage{dsfont}
\usepackage{cancel}
\usepackage{enumitem}

\usepackage[capitalise, noabbrev]{cleveref}
\usepackage{multirow}
\hypersetup{colorlinks=true,linkcolor=blue,citecolor=blue,urlcolor=blue}
\usepackage[margin=2cm]{geometry}

\usetikzlibrary{decorations.pathmorphing}
\usepackage{listings}
\usepackage{color}

\usepackage{tikz,pgfplots}
\usetikzlibrary{quantikz}
\usetikzlibrary{arrows,shapes,positioning}
\usetikzlibrary{decorations.markings}

\theoremstyle{plain}

\newtheorem{theorem}{Theorem}

\newtheorem{prop}[theorem]{Proposition}
\newtheorem{obs}[theorem]{Observation}

\newtheorem{claim}[theorem]{Claim}

\theoremstyle{remark}

\theoremstyle{definition}

\newcommand{\loc}{L}
 
\definecolor{codegreen}{rgb}{0,0.6,0}
\definecolor{codegray}{rgb}{0.5,0.5,0.5}
\definecolor{codepurple}{rgb}{0.58,0,0.82}
\definecolor{backcolour}{rgb}{0.95,0.95,0.92}

\lstdefinestyle{mystyle}{
    backgroundcolor=\color{backcolour},   
    commentstyle=\color{codegreen},
    keywordstyle=\color{magenta},
    numberstyle=\tiny\color{codegray},
    stringstyle=\color{codepurple},
    basicstyle=\footnotesize,
    breakatwhitespace=false,         
    breaklines=true,                 
    captionpos=b,                    
    keepspaces=true,                 
    numbers=left,                    
    numbersep=5pt,                  
    showspaces=false,                
    showstringspaces=false,
    showtabs=false,                  
    tabsize=2
}
 
\usepackage{tikz}
\usetikzlibrary{arrows}

\DeclareMathOperator{\diag}{diag}

\usepackage{array}
\usepackage{rotating}
\usepackage{tabularx}
\usepackage{geometry}
{\protect}

\newcommand{\id}{\mathbf{1}}

\renewcommand{\emph}[1]{{\it #1}}
\newcommand{\tr}{\mathrm{Tr}}
\newcommand{\ketbra}[2]{|{#1}\rangle \! \langle {#2}|}

\def\E{ {\cal E} }

\newcommand{\dyad}[2]{| #1\rangle \langle #2|}

\newcommand{\norm}[1]{\left\| #1 \right\|}

\usepackage {tikz}
\usetikzlibrary {positioning}

\usepackage{setspace}
\usepackage{orcidlink}

\newcommand{\distquantum}[1]{\mathcal N_{#1}}
\newcommand{\distclassical}[1]{\nu_{#1}}

\usepackage{xcolor}

\newcommand*{\colorboxed}{}
\def\colorboxed#1#{%
  \colorboxedAux{#1}%
}
\newcommand*{\colorboxedAux}[3]{%
  \begingroup
    \colorlet{cb@saved}{.}%
    \color#1{#2}%
    \boxed{%
      \color{cb@saved}%
      #3%
    }%
  \endgroup
}

\makeatletter
\renewcommand*\env@matrix[1][\arraystretch]{%
  \edef\arraystretch{#1}%
  \hskip -\arraycolsep
  \let\@ifnextchar\new@ifnextchar
  \array{*\c@MaxMatrixCols c}}
\makeatother

\begin{document}

\title{
Channel nonlocality under decoherence 
}

\author{Albert Rico${}^{\orcidlink{0000-0001-8211-499X}}$}
\affiliation{
Faculty of Physics, Astronomy and Applied Computer Science, Institute of Theoretical Physics, Jagiellonian University,
30-348 Krak\'{o}w, 
Poland}
\author{Mois\'es Bermejo Mor\'an${}^{\orcidlink{0000-0003-1441-0468}}$}
\affiliation{
Faculty of Physics, Astronomy and Applied Computer Science, Institute of Theoretical Physics, Jagiellonian University,
30-348 Krak\'{o}w, 
Poland}
\affiliation{
Laboratoire d’Information Quantique, Université libre de Bruxelles, Belgium
}
\author{Fereshte Shahbeigi${}^{\orcidlink{0000-0001-9991-6112}}$}
\affiliation{
Faculty of Physics, Astronomy and Applied Computer Science, Institute of Theoretical Physics, Jagiellonian University,
30-348 Krak\'{o}w, 
Poland}
\affiliation{RCQI, Institute of Physics, Slovak Academy of Sciences, D\'{u}bravsk\'{a} cesta 9, 84511 Bratislava, Slovakia}
\author{Karol \.{Z}yczkowski${}^{\orcidlink{0000-0002-0653-3639}}$}
\affiliation{
Faculty of Physics, Astronomy and Applied Computer Science, Institute of Theoretical Physics, Jagiellonian University,
30-348 Krak\'{o}w, 
Poland}
\affiliation{
Center for Theoretical Physics, Polish Academy of Science,
02-668 Warszawa, Poland}

\date{\today}

\begin{abstract}

The implementation of realistic quantum devices requires a solid understanding of the nonlocal resources present in quantum channels, and the effects of decoherence on them. 
Here we quantify nonlocality of bipartite quantum channels and identify its component resisting the effects of dephasing noise.
Despite its classical nature, we demonstrate that the latter plays a relevant role in performing quantum protocols, such as state transformations and quantum coding for noisy communication.
In the converse direction, we show that  simulating certain stochastic processes with quantum channels undergoing decoherence has a communication advantage with respect to their classical  simulation.

\end{abstract}

\maketitle

{\em Introduction.\,\,} 
The phenomenon of Bell nonlocality lies at the core of quantum theory, as quantum distributions cannot always be explained by hidden variables assuming local realism~\cite{bell1964einstein, bell1966problem}. 
Nonlocality constitutes a resource for device-independent protocols in cryptography and information theory, such as public key distribution schemes, conjugate coding, and 
randomness verification 
~\cite{wiesner1983conjugate,pironio2010random,bennett2014quantum}. 
Nonlocal resources can be present both in the quantum state describing a system~\cite{BrunnerReviewBellNonloc_2014}, and in a quantum operation manipulating such state~\cite{ChSteerNLbeyond_Hoban2018Games,ResNCchAssem_Zjawin2023Games,CorrCompMeas_Selby2023Games,rosset-type-independent-prl,schmid2020type}. 

In practice, both the storage of a quantum state and its manipulation through a quantum channel suffer from the phenomenon of {\em decoherence}~\cite{peres1997quantum,Nielsen_Chuang_2010}: 
the interaction with the environment 
can fully or partially spoil quantum properties~\cite{schlosshauer2007quantum,Breuer-opensysytems,wiseman2009qmeascontrol}. Therefore, successful realization of noisy quantum devices demands a deep understanding of quantum properties under decoherence~\cite{PRXQCriticallityUnderDeco_Lee2023,devetak2005capacitydeco,lee2024symmetryprotectedtopologicalphasesDeco,lyons2024understandingstabilizercodeslocalDeco}. For this purpose, a well-established model of decoherence is the dephasing noise $\mathcal{D}$, which damps the off-diagonal terms in a density matrix while keeping intact the terms diagonal in a preferred  basis~\cite{ResCoh_Winter2016}. In the worst-case scenario, complete decoherence leads to vanishing of the off-diagonal terms. This reduces a quantum state to a probability vector~\cite{devetak2005capacitydeco,DArrigo_2007_CapDepChann,Bradler2010TroffCapHadChann,ResCoh_Winter2016}  and a quantum channel to a stochastic map described by a stochastic matrix~\cite{Puchala_Dephasing2021}.   



In recent years, increasing attention has been paid to understanding the nonlocal properties of quantum operations. When no shared resources are allowed, the 
operator Schmidt decomposition of a bipartite unitary gate~\cite{NDD+03QuantDynPhysRes_Nielsen2003,MKZ13UnitQGatesEntUnist_Musz2013} or a quantum channel~\cite{MAH+24QTensor-ProdDecChoi_Mansuroglu2024} determines whether it can be performed locally. On the other hand, linear tests for nonlocality of states~\cite{BuscemiNonlocHV_2012,BrunnerReviewBellNonloc_2014} can be lifted to the level of channels~\cite{schmid2020type}. There, the probability distributions
obtained by local quantum measurements~\cite{ChSteerNLbeyond_Hoban2018Games,schmid2020type,CorrCompMeas_Selby2023Games,OurLong_2024} reveal whether a bipartite channel can be implemented locally with shared randomness. 
While these methods provide a qualitative characterization, a quantitative analysis is desired for a further step in understanding both the nonlocal resources of channels and the effects of decoherence on them. 

In this work we quantify the nonlocal resources of bipartite channels with several measures, and their corresponding components remaining when complete decoherence occurs. 
The latter determine the nonlocal transformations that a channel can induce on the computational basis, e.g. bit pairs. 
Using a semidefinite programming hierarchy for bilinear optimization~\cite{berta2021semidefinite}, we evaluate the role of nonlocal resources in performing quantum protocols. 
Despite accounting for a classical resource, the nonlocality resisting decoherence is found to be required for optimal performance. 
In a similar direction, we compare two different protocols to  simulate discrete Markov chains~\cite{Breuer-opensysytems}, namely using standard classical methods with respect to using a quantum channel that undergoes complete decoherence. 
We show that the latter gives advantage in the communication needed.

{\em Preliminaries.\,\,} 
Our work 
relies on the bridge between two seemingly different structures~\cite{FormSteeringLM_Sainz2018Games,ChSteerNLbeyond_Hoban2018Games,CorrCompMeas_Selby2023Games,ResNCchAssem_Zjawin2023Games,ResThNCcomcauAssem_Zjawin2023Games,OurLong_2024}: the polytope of bipartite distributions~\cite{BrunnerReviewBellNonloc_2014}, and the set of bipartite quantum channels~\cite{Geller_2014,BGNP01causal,schmid2021postquantum,Eggeling_Semicausal2002}. A bipartite distribution $p(ab|xy)$ is \emph{local} (\loc{}) if it decomposes as
$p(a,b|x,y)=\sum_\lambda p_\lambda \, p(a|x,\lambda) p(b|y,\lambda)\,$ 
with $p_\lambda \geq 0$ and $\sum_\lambda p_\lambda = 1$,
and nonlocal otherwise~\cite{BrunnerReviewBellNonloc_2014}. Distributions obtained from bipartite quantum states $\rho_{AB}$ via Born's rule,
$ p(a,b|x,y) = \tr(P^{a|x}\otimes Q^{b|y}\rho_{AB}) $
for local measurement effects $P^{a|x}$ and $Q^{b|y}$, are called \emph{quantum} (Q). Nonlocal resources present in quantum states~\cite{bell1964einstein,werner1989quantum} cannot increase under local operations with shared randomness (LOSR)~\cite{BuscemiNonlocHV_2012,schmid2020type,rosset-type-independent-prl,Liu_ResTherQChan2020},
$\Phi_{AB}=\sum_\lambda p_\lambda\Phi_A^\lambda\otimes\Phi_B^\lambda\,$. 
Bipartite channels outside the set of LOSR are called {\em nonlocal} \cite{OurLong_2024}. The set of LOSR operations is preserved under local superchannels with shared randomness, 
$\sum_\lambda p_\lambda\Xi^\lambda_A\otimes\Xi^\lambda_B$, which are the free operations in the resource theory of nonlocality of channels~\cite{schmid2020type}. Further details about the main structure and relevant subsets of bipartite channels and distributions are given in Appendix~\ref{app:NLchannels&distributions}.

{\em Nonlocality of quantum and classical channels.\,\,}
The nonlocal resources available in a quantum channel $\E: A_0 B_0 \to A_1 B_1$ will be treated as a convex resource theory \cite{regula2017convex}, where LOSR operations are the free elements. Therefore, we quantify the nonlocal resources of a channel as
\begin{equation}\label{eq:ChanNLdef}
\distquantum{\Delta} (\E) = \min \{\Delta(\E, \Phi) : \Phi \in \operatorname{LOSR}\}
\end{equation}
for a suitable non-negative functional $\Delta$. 
We require that $\Delta$ is continuous, convex, non-degenerate and contractive under superchannels so that $\mathcal N_\Delta$ attains the minimum, is convex, faithful and monotonic under local superchannels respectively (see Appendix~\ref{app:ProofNQmonotone}). 

Examples are $\mathcal{N}_1$ defined from the trace distance  $\Delta(\E, \Phi) = \norm{J^\E - J^\Phi}_1$ with $\norm{X}_1= \tr \sqrt{XX^\dagger}$, which admits semidefinite programming bounds for unitary channels as shown in Appendix~\ref{app:ProofNQmonotone}; $\mathcal{N}_\diamond$ from the diamond distance $\Delta(\E, \Phi) = \norm{\E - \Phi}_\diamond$ with $\norm{\E}_\diamond  = \sup\{ \norm{\E \otimes \id (X)}_1 : \norm{X}_1 \leq 1\}$; and $\mathcal{N}_H$ from the relative entropy $\Delta(\E,\Phi)=H(J^\E|J^\Phi) = \tr J^\E (\log J^\E - \log J^\Phi)$, 
where $J^\E$ 
is the unnormalized four-partite Choi state~\cite{choi1975completely, jamiolkowski1972linear}. Faithful monotones can similarly be obtained from Eq.~\eqref{eq:ChanNLdef} using quantum R\'enyi divergences \cite{van2005statistical, lennert2013onquantum, fawzi2021defining}. 
Physically, $\distquantum{\diamond}(\E)$ corresponds to the maximal success probability of discriminating $\E$ with respect to the least distinguishable LOSR operation $\Phi$~\cite{HelstromDisting1969,Superchannels_Gour2019}. 
This can be bounded with $\distquantum{1}(\E)$ 
through $\distquantum{1}(\E) \leq \distquantum{\diamond}(\E) \leq \dim A_0 B_0 \distquantum{1}(\E)$. 
Similarly, $\distquantum{H}(\E)$ measures the smallest relative entropy between the Choi matrices of $\E$ and some LOSR channel $\Phi$, which is a frequent measure for their distinguishability~\cite{Superchannels_Gour2019}.

\begin{figure}[tbp]
\includegraphics[width=0.45\textwidth]{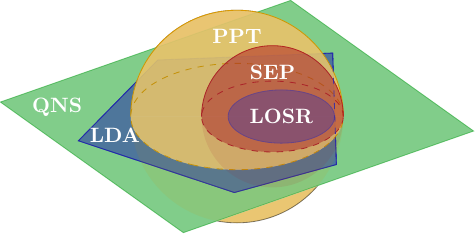}
    \caption{{\bf Subsets of bipartite channels and their cross sections.} The orange convex set (PPT) contains PPT-preserving channels~\cite{Leung_NSCodesSDP_2015}. 
    The green hyperplane represents the quantum nonsignaling constraints (QNS) in Eqs.~\eqref{eq:qNSab} and \eqref{eq:qNSba}, and contains the purple set of local operations with shared randomness (LOSR). The intersection between QNS and  separable channels~\cite{chitambar2014everything} (SEP) depicted in red outer-approximates LOSR through the first level of the hierarchy~\cite{berta2021semidefinite} given in Eqs.~\eqref{eq:appCP}-\eqref{eq:appCONS2}. This approximation can be strengthened with linear constraints imposing a local decoherent action (LDA), which give the polytope depicted in blue.} 
    
    \label{fig:SetsSepChan}
\end{figure}

Consider a bipartite channel $\mathcal{E}_{AB}$ which undergoes dephasing noise $\mathcal{D}(\dyad{i}{j})=\dyad{i}{j}\delta_{ij}$ both in its input and output states, $\E^\mathcal{D}:=\mathcal{D}\circ\E\circ\mathcal{D}$, where $\delta_{ij}$ is the Kronecker delta and $\circ$ denotes channel concatenation. Then, one can describe the effective action of $\E^\mathcal{D}$ by the {\em decoherent action} (or classical action) of $\E$, which is a classical channel acting on digits~\cite{InfProcGPTClasChan_Barrett2007,InterconvNLcorrClasChan_Jones2005,OprFramewNLClasChan_Gallego2012,NLRTmeasClasChan_deVicente2014} given by a bipartite stochastic matrix $S^\E$ with entries
~\cite{KamilCohQChan_2018,shahbeigi-log-convex,Puchala_Dephasing2021,QAdvSimStoch_Kamil2021,QEmbStoch_Fereshte2023,OurLong_2024}
\begin{equation}\label{eq:CAchannel}
 S_{ab,xy}^{\mathcal{E}}=\bra{ab}\mathcal{E}(\dyad{xy}{xy})\ket{ab}\,,
\end{equation}
where indices $(a,b,x,y)$ label the canonical basis of $A_1 B_1 A_0 B_0$. The concatenation of the columns of $S^\E$ coincides with the diagonal of the Choi matrix of the channel $\E$, which is preserved under complete decoherence. As bipartite distributions correspond with bipartite stochastic matrices, the decoherent action~\eqref{eq:CAchannel} defines probabilities $p(a,b|x,y)=S_{ab,xy}$ of obtaining outputs $a,b$ conditioned to inputs $x,y$ through a measurement protocol performed in the computational basis~\cite{ChSteerNLbeyond_Hoban2018Games,schmid2020type,OurLong_2024}.
Correspondingly, we denote as {\em decoherent nonlocality} of a channel $\E$ the nonlocal resources present in the distribution defining its decoherent action $S^\E$.

Nonlocality of bipartite conditional distributions (or classical channels) can be quantified in several ways~\cite{ViolationsRobustnessWN_Kaszlikowski2000,QNLin2-3levRobustnessWN_Acin2002,perez2008unbounded,RobustnessLNoise_Massar2002,deVicente_2014}. For any continuous, convex, non-degenerate and contractive non-negative functional $\delta$ we introduce the measure
\begin{equation}\label{eq:MeasNLtrDist}
    \distclassical{\delta}(S)=\min \{ \delta(S,T) : T \in \operatorname{L} \}\,. 
\end{equation}
Note that $\nu_\delta(S^\E)$ quantifies the decoherent nonlocality of a bipartite channel $\E$. Indeed, each measure $\distclassical{\delta}$ arises from a measure $\mathcal{N}_\Delta$ by taking $\delta(S,T) = \Delta(\E,\Phi)$ for channels $\E = \mathcal{D}\circ\E\circ\mathcal{D}$ and $\Phi = \mathcal{D}\circ\Phi\circ\mathcal{D}$ with decoherent actions $S$ and $T$. A stronger correspondence between Eqs.~\eqref{eq:ChanNLdef} and \eqref{eq:MeasNLtrDist} holds:
\begin{prop}\label{prop:NcvsNq}
The following relations hold between each measure $\distquantum{\Delta}$ of channel nonlocality and its corresponding measure $\distclassical{\delta}$ of decoherent nonlocality.
\begin{enumerate}[label=(\roman*)]
\item $\distquantum{\Delta}(\mathcal D \circ \E \circ \mathcal D) =  \distclassical{\delta}(S^\E)$
\item $\distquantum{\Delta}(\E)\geq\distclassical{\delta}(S^\E)$.
\end{enumerate}
\end{prop}
\noindent Proposition~\ref{prop:NcvsNq} allows for bounding the quantum nonlocal resources of a channel from its nonlocality resisting decoherence, which is arguably much simpler to estimate both in theory and experiments. Geometrically, item (i) holds because the dephasing channel $\mathcal{D}$ induces a cross section in the space of channels which is also a projection, and item (ii) follows from (i). A detailed proof is given in Appendix~\ref{app:NcvsNq}.

Examples of Proposition~\ref{prop:NcvsNq} are 
sketched below and discussed in Appendix~\ref{app:NcvsNq}. Complete dephasing reduces $\distquantum{1}$ to the measure $\distclassical{1}$ defined from $\delta_1(S, T) = \norm{S - T}_{\ell_1} = \sum_{ij} |S_{ij} - T_{ij}| $. This coincides with the measure of nonlocality for bipartite distributions introduced in~\cite{BritoQuantBellNLtrDist_2018}, whose value in settings with two measurements (corresponding to channels with two dimensional local input spaces) is determined by the maximum value of certain Bell functionals on the subset~\cite{GeoDecompBellPolyt_Bierhorst2016,BritoQuantBellNLtrDist_2018}. For the relative entropy, $\distquantum{H}$ reduces to $\distclassical{H}$ defined from the classical relative entropy, $\delta_H(S, T) = H(S|T) = \sum_{ij} S_{ij} (\log S_{ij} - \log T_{ij})$. Last, the diamond measure $\distquantum{\diamond}$ reduces under the decoherence to $\distclassical{\diamond}$ defined from $\delta_\diamond(S, T) = \norm{S - T}_{11}$, where 
$\norm{S}_{11} = \sup \left\{ \sum_i \left| \sum_j S_{ij} x_j \right| : \sum_j |x_j| \leq 1\right\}$
is the operator norm between classical sequence spaces $S : \ell_1(A_0 B_0) \to \ell_1(A_1 B_1)$. Similarly, quantum R\'enyi divergences reduce under decoherence to classical ones \cite{van2005statistical, lennert2013onquantum, fawzi2021defining}.

{\em Nonlocality resisting decoherence.\,\,} The measures $\distclassical{\delta}$ of decoherent nonlocality can be seen as the amount of nonlocal resources that remain after complete decoherence occurs. In particular, these serve as a witness for nonlocality of the channel $\E$ in hand: if $\distclassical{\delta}(S^\E) > 0$ ($S^\E\not\in$ \loc{}), then $\distquantum{\Delta}(\E) > 0$ ($\E\not\in$ LOSR)~\cite{ChSteerNLbeyond_Hoban2018Games,OurLong_2024}.  Consider for example the quantum channel $\E^U_{CN}(\rho)=U_{CN}\rho U_{CN}^\dag$ defined by the controlled-not gate, $ U_{CN}=\sum_{i=0}^1\dyad{i}{i}\otimes\sigma_X^i$.
One verifies that its decoherent action satisfies $S^{U_{CN}}=U_{CN}\odot U_{CN}=U_{CN}$, where $\odot$ is the Schur (entry-wise) product~\cite{OurLong_2024}. This means that after dephasing occurs, the action of the controlled-not gate in two classical bits remains untouched. In fact, $S^{U_{CN}}$ defines an extremal one-way signaling bipartite distribution and maximizes $\distclassical{1}$.

This raises the question of whether the nonlocal resources resisting decoherence, which constitute the decoherent nonlocality, are related with the performance of a channel in quantum protocols which do not suffer from decoherence. The convergent hierarchy detailed in Appendix~\ref{app:hierarchyLOSR} will be used to tackle this question numerically with semidefinite programming.
Although this hierarchy is complete, only the first levels are available in practice and further linear constraints are desired to solve realistic problems. 
Here we will simply add the following linear constraint at each level,
\begin{equation}\label{eq:appLochier}
    \diag(J_{A_0B_0A_1B_1})\in\text{\loc{}},
\end{equation}
which requires that the probability distribution given by the decoherent action lies in the local polytope (Fig.~\ref{fig:SetsSepChan}) and thus enforces $\distclassical{\delta}=0$. Straightforward examples of the strengthening of Eqs.~\eqref{eq:appCP}-\eqref{eq:appCONS2} through Eq.~\eqref{eq:appLochier} are diagonal Choi matrices given by $J^\E=\sum_{xyab}p(ab|xy)\dyad{xyab}{xyab}$ where $p=S^\E$ is nonsignaling (see Appendix~\ref{app:hierarchyLOSR}).

\begin{figure*}[tbp]
\centering
\includegraphics[width=\textwidth]{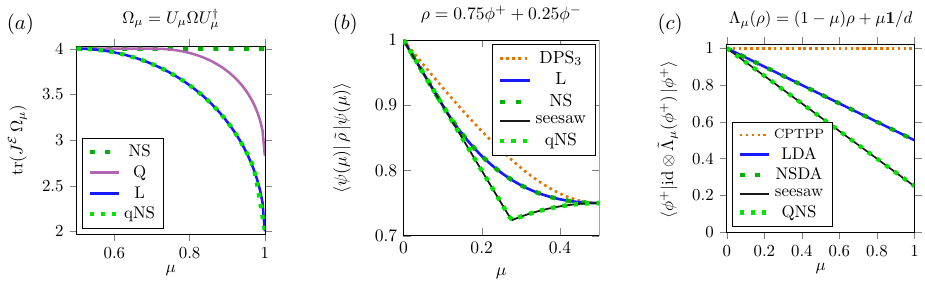}
\caption{{\bf Optimal performance of bipartite channels $\E$ in restricted subsets at quantum protocols with parameter $\mu$}.
Different bounds are obtained when 
$S^\E$ defines local (LDA), quantum (NPA$_3$),
nonsignaling (NSDA), and general distributions, where NPA$_3$ denotes third-level approximation with~\cite{navascues2008convergent}. 
Bounds for LOSR are obtained when those obtained with quantum-nonsignaling channels (QNS) and with a seesaw algorithm coincide. We depict the maximum value of: {(a)} The Bell functional $\Omega_\mu$ in Eq.~\eqref{eq:OptCohWit}. 
{(b)} The fidelity of $\Tilde{\rho}=\E(\rho)$ to 
a target pure state $\ket{\psi_\mu}$, 
where DPS$_3$ approximates SEP at the third level of~\cite{doherty2004complete,Navascues_PPT_DPS_2009}.
{(c)} The fidelity of the depolarizing channel $\Lambda_\mu$ 
assisted 
according to Eq.~\eqref{eq:assisted_communication} and Fig.~\ref{fig:AssistedCommDiagram}.
General superchannels 
(CPTPP)~\cite{Supermaps_Chiribella2008,Superchannels_Gour2019} allow for unit fidelity through $\Xi[\Lambda_\mu]=\text{id}$.
}\label{fig:plot-tasks}
\end{figure*}

{\em Bell functionals on quantum channels.\,\,}
We will introduce functionals on Choi matrices that reach different maximum values for different subsets of bipartite channels. 
This generalizes through Proposition~\ref{prop:NcvsNq} the analogous property of $\distclassical{\delta}$ for bipartite distributions. 
Given a linear functional $\Omega$ on the space of Choi matrices for channels $\E_{AB}$~\cite{OurLong_2024}, consider the maximum value that can be obtained with quantum channels whose decoherent action belongs to a preferred subset $\mathcal L$,
\begin{equation}\label{eq:OptCohWit}
\gamma_{\mathcal L}(\Omega) :=\max \{ \tr(J^\E\Omega) : S^\E \in \mathcal L \}\, .
\end{equation}
Let us fix $\Omega_\mu=U_\mu\Omega U_\mu^{\dag}$, which is a coherified version of the diagonal functional $\Omega=\sum_{x,y,a,b=0}^1 (-1)^{a+b}(-1)^{xy}\dyad{xyab}{xyab}$ used in~\cite{OurLong_2024} to evaluate the Clauser-Horne-Shimony-Holt inequality~\cite{CHSHineq1969} on the decoherent action. Nonzero off-diagonal terms are given by the unitary $U_\mu=\big (\ket{\phi^+_\mu}\bra{00}+\ket{\psi^+_\mu}\bra{01}+\ket{\psi^-_\mu}\bra{10}+\ket{\phi^-_\mu}\bra{11}\big )_{A_0B_0}\otimes\mathbb{1}_{A_1B_1}$ where $\ket{\psi^{\pm}_\mu}=X^{\pm}(\sqrt{\mu}\ket{01}\pm\sqrt{1-\mu}\ket{10})\sqrt{2}$ with $X^{\pm}=(\sigma_X\otimes\sigma_X)^{(\pm1-1)/2}$. 
The states $\ket{\psi^{\pm}_\mu}$ and $\ket{\phi^{\pm}_\mu}$ are entangled for $\mu\neq 0,1$ and form an orthonormal basis for all $\mu$, e.g. the computational basis for $\mu =1$ and the Bell basis for $\mu=1/2$.  Fig.~\ref{fig:plot-tasks} (a) shows that the optimal value of Eq.~\eqref{eq:OptCohWit} depends on the nonlocality of the decoherent action throughout the whole range $\mu\in(1/2,1]$:
a continuous gap exists between the sets of channels whose decoherent action is in the set of local, quantum (an outer approximation~\cite{navascues2008convergent}), and nonsignaling distributions. 

{\em State interconversion.\,\,} 
Consider the fidelity of converting a Bell state which suffered from the unitary noise $\id\otimes\sigma_Z$ with probability $1-p$, $\rho=p\phi^++(1-p)\phi^-$, into a pure state with entanglement entropy $E=\mu\log_2\mu^{-1}+(1-\mu)\log_2(1-\mu)^{-1}$, namely $\ket{\psi_\mu}=\sqrt{\mu}\ket{00}+\sqrt{1-\mu}\ket{11}$ up to local unitaries. The figure of merit is the fidelity between the transformed state by a channel $\mathcal E$ and the target state $\ket{\psi_\mu}$,
\begin{equation}
   \mathcal{F}= \bra{\psi_\mu}\mathcal{E}(\rho)\ket{\psi_\mu}\,.
   \label{eq:fidelity_transformation}
\end{equation}

Standard physical scenarios require that transformations can be done through quantum operations with a limited amount of resources
~\cite{ConcEnt_Bennet1996,NielsenMajor_1999,TransfRev_Horodecki2003,ResThThermo_Brandao2013,QuantifCoh_Baumgratz2014,ResCoh_Winter2016}. Here we limit the nonlocal resources by adding to the hierarchy detailed in Appendix~\ref{app:hierarchyLOSR} further constraints to bound selected subsets of quantum channels (see Appendix~\ref{app:NLchannels&distributions}). Fig.~\ref{fig:plot-tasks} (b) shows that 
the nonlocality present in the decoherent action of the channel $\E$ (Eq.~\eqref{eq:CAchannel}) 
can increase the fidelity of the transformation in Eq.~\eqref{eq:fidelity_transformation}.

{\em Quantum codes.\,\,}
\begin{figure}[tbp]
    \centering
    \includegraphics{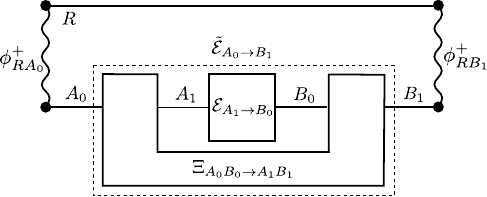}
    \caption{{\bf Assisted communication.} In a standard communication protocol, one part $A_1$ of a Bell state $\phi^+$ is sent to another part $B_0$ through a noisy channel $\mathcal E$ with input $A_1$ and output $B_0$. In an assisted communication protocol, the noisy channel is corrected to $\Tilde{\E}_{A_0\rightarrow B_1}=\Xi_{A_0B_0\rightarrow A_1B_1}(\E_{A_1\rightarrow B_0})$ (dashed outline) by a superchannel $\Xi$. The fidelity of the assisted channel is then the fidelity between the input Bell state and its output image, $\mathcal{F}=\bra{\phi^+}\Tilde{\mathcal E}(\phi^+)\ket{\phi^+}$.}
    \label{fig:AssistedCommDiagram}
\end{figure}
Consider the task of sending information from Alice to Bob through a noisy channel $\mathcal E$. The assisted fidelity of sending one part of a maximally entangled state, 
\begin{equation}
    \mathcal{F}=\tr\Big (\phi^+_{RB_1}\cdot\big (\text{id}_{R}\otimes\Tilde{\mathcal E}_{A_0\rightarrow B_1}\big )(\phi^+_{RA_0})\Big )\,,
    \label{eq:assisted_communication}
\end{equation}
quantifies how well information is transmitted through $\Tilde{\mathcal E}_{A_0\rightarrow B_1}:=\Xi\left(\mathcal E_{A_1\rightarrow B_0}\right)$, which is the transformed version of $\mathcal E_{A_1\rightarrow B_0}$ by the superchannel $\Xi$ (Fig.~\ref{fig:AssistedCommDiagram})~\cite{Leung_NSCodesSDP_2015,BrandaoEAcodesCantNoise_2011,Capacity_Wang2012,WangSDPboundsCapacity_2018,RelEntrCapacity_Mosonyi2009,NoiseChCodesCapacity_Renes2011}. The superchannel $\Xi$ is equivalent to a bipartite channel $\Xi_{A_0B_0\rightarrow A_1B_1}$ whose Choi matrix satisfies $J\geq 0$ with the marginals $J_{A_0B_0}=\id_{A_0B_0}$ and the one-way nonsignaling constraint $J_{A_0A_1B_0}= J_{A_0A_1}\otimes\id_{B_0}/d_{B_0}$~\cite{BGNP01causal,Superchannels_Gour2019}. These constraints alone are compatible with $\mathcal{F}=1$ since they admit a contractive superchannel $\Xi$ such that $\tilde{\mathcal E}=\text{id}$ is the noiseless identity channel. However, physically realistic scenarios can be modelled by imposing extra constraints to $\Xi$, leading to nontrivial bounds~\cite{BrandaoEAcodesCantNoise_2011,WinterBoundsQComm_2015,Leung_NSCodesSDP_2015,WangSDPboundsCapacity_2018,ChanFidSymm_Chee2024}. Bounds for the depolarizing channel $\Lambda_\mu$ with mixing parameter $\mu$ obtained through the strengthened hierarchy of Appendix~\ref{app:hierarchyLOSR} are depicted in Fig~\ref{fig:plot-tasks}~(c). The results of the three protocols depicted in Fig~\ref{fig:plot-tasks} imply the following. 
\begin{prop}
The nonlocal resources of a quantum channel that remain available when coherences are lost (quantified by $\distclassical{\delta}$) are required for optimal performance of coherent quantum protocols.
\end{prop}

{\em Quantum  simulation of stochastic matrices.\,\,}
Stochastic matrices, also called classical channels~\cite{InterconvNLcorrClasChan_Jones2005,InfProcGPTClasChan_Barrett2007,OprFramewNLClasChan_Gallego2012,NLRTmeasClasChan_deVicente2014,QRTChitambar2019}, model the evolution of discrete probabilistic systems and are the building blocks of certain stochastic processes like discrete Markov chains~\cite{Breuer-opensysytems}. 
Although these are classical, it has been shown that the quantum simulation of stochastic processes provides advantage in certain resources such as memory and space-time cost with respect to their classical  simulation~\cite{QAdvSimStoch_Kamil2021,QEmbStoch_Fereshte2023}. 
A common tool for this purpose is the decoherent action, using the fact that complete decoherence preserves the action of a channel on classical systems. 
 
Here we will follow a similar approach: one verifies that two parties that share a common classical source of randomness but cannot communicate, can only  simulate bipartite stochastic matrices of the form
\begin{equation}\label{eq:losr-stochastic-matrix}
    S_{AB}=\sum_\lambda p_\lambda T_A^\lambda\otimes R_B^\lambda\, 
\end{equation}
where 
$T_A^\lambda$ and $R_B^\lambda$ are stochastic matrices. This is because only local bipartite distributions can be obtained in this case~\cite{ComCostSimBell_Toner2003}.
In contrast, larger sets of bipartite stochastic matrices can be  simulated using shared quantum resources undergoing noise:

\begin{prop}\label{prop:SimulStoch}
There are stochastic processes that require communication to be  simulated classically, but can be  simulated quantumly without communication through a channel undergoing decoherence.


\end{prop}
To see this, one can use quantum channels which are implemented by using a shared entangled state and applying local operations. In Appendix~\ref{app:SimulStoch} we explain how stochastic matrices corresponding to nonlocal distributions can be obtained when the channel in hand undergoes dephasing noise, in which case measurements can only be performed in the computational basis. 
Proposition~\ref{prop:SimulStoch} implies that quantum advantage in communication remains accessible when coherences are lost in the inputs and outputs of the channel in use. This allows us to use the techniques of~\cite{QAdvSimStoch_Kamil2021}, where quantum advantage was shown in noisy channels, in the spirit of~\cite{dense-coding}, where entanglement reduces communication needed to implement a classical channel. 


{\em Conclusions.\,\,} 
We introduce several measures of nonlocality 
for bipartite channels with desirable properties, and identify their component remaining after decoherence.
Despite accounting for classical resources at the level of channels, we show that 
nonlocality resisting decoherence is required for 
certain quantum protocols: maximizing Bell functionals for channels, state interconversion fidelity, and assisted codes fidelity for noisy channel communication. 
This analysis also demonstrates a gap between different subsets of bipartite quantum channels. 
These results provide significant insight into the nonlocal properties resisting dephasing noise and their role in quantum implementations: Bell nonlocality of correlations is understood as the classical component of channel nonlocality, that can be implemented with fully classical devices.



On a more conceptual level, nonlocality of channels (quantum or classical) accounts for the resources needed for their implementation in distant labs.
These resources include classical communication or shared entanglement. Combining these two, any quantum operation can be locally implemented, while classical communication is enough to locally realize any classical operation. 
Our analysis suggests a natural way to  simulate bipartite stochastic matrices, from bipartite quantum channels which undergo complete decoherence. 
We show that this quantum  simulation method has a communication advantage with respect to classical  simulation.
Our results open an avenue to understanding the tradeoff between entanglement and communication costs of implementing a bipartite quantum operation. 


\bigskip

{\em Acknowledgements:\,\,} We are thankful to 
Remigiusz Augusiak, 
Francesco Buscemi,
Dariusz Chru\'{s}ci\'{n}ski, 
Felix Huber, 
Kamil Korzekwa, 
Miguel Navascu\'{e}s, 
Roberto Salazar, 
Anna Sanpera and 
Mario Ziman 
for valuable comments and for pointing to relevant references. Support by the Foundation for Polish Science through TEAM-NET project  POIR.04.04.00-00-17C1/18-00
and by NCN QuantERA
Project No. 2021/03/Y/ST2/00193 is gratefully acknowledged.  FS acknowledges projects DESCOM VEGA 2/0183/2 and DeQHOST APVV-22-0570.


\appendix
\setcounter{secnumdepth}{1}

\section{Subsets of bipartite channels and bipartite distributions}\label{app:NLchannels&distributions}

Consider a bipartite channel $\mathcal{E}_{AB}$ acting as $\E_{AB}(\rho_{A_0B_0})=\sigma_{A_1B_1}$, with Choi matrix $J_{A_0B_0A_1B_1}$~\cite{choi1975completely,jamiolkowski1972linear}. In the set of bipartite quantum channels, one identifies the following distinguished subsets: 

{\em Local operations assisted by shared randomness} (LOSR), satisfying
\begin{equation}\label{eq:appLOSR}
 \mathcal{E}_{AB}=\sum_\lambda p_\lambda \mathcal{E}_A^\lambda\otimes\mathcal{E}_B^\lambda\,,
\end{equation}
where $p_\lambda\geq 0$, $\sum_\lambda p_\lambda=1$ and $\mathcal{E}_{A(B)}^\lambda$ are quantum channels. These are the free operations in the resource theory of Bell nonlocality~\cite{Geller_2014,deVicente_2014}. Channels outside the convex set of LOSR are termed \emph{nonlocal}.

{\em Local operations with shared entanglement} (LOSE), which can be applied locally with a shared quantum state $\tau_{EF}$~\cite{BGNP01causal,schmid2021postquantum},
\begin{equation}\label{eq:Localizable}
     \!\!\mathcal{E}(\rho_{A_0B_0})\!=\!\tr_{R_1S_1} \big(\mathcal{E}_{AR}\otimes\mathcal{E}_{BS}(\rho_{A_0B_0}\!\otimes\tau_{R_0S_0})\big).
 \end{equation}
{\em Quantum nonsignaling (QNS)} operations, whose Choi matrix satisfies
  \begin{subequations}
  \begin{align}
  J^{\mathcal{E}}_{A_0B_0B_1}&=\id_{A_0}\otimes J^{\mathcal{E}}_{B_0B_1}/d_{A_0},\label{eq:qNSab}\\ J^{\mathcal{E}}_{A_0A_1B_0}&= J^{\mathcal{E}}_{A_0A_1}\otimes\id_{B_0}/d_{B_0}\,\label{eq:qNSba}
  \end{align}
  \end{subequations}
where we denote the partial trace by omitting the traced subsystems.

In the polytope of bipartite distributions (bipartite stochastic matrices), one distinguishes the following subsets: \emph{local} distributions (\loc{})~\cite{bell1966problem,BrunnerReviewBellNonloc_2014}, decomposing as
\begin{equation}
\label{eq:ProbSep}
    p(a,b|x,y)=\sum_\lambda p_\lambda \, p(a|x,\lambda) p(b|y,\lambda)\,,
\end{equation}
where $p_\lambda\geq 0$, $\sum_\lambda p_\lambda=1$; {\em quantum} distributions (Q), obtained as
\begin{equation}\label{eq:ProbBorn}
 p(a,b|x,y) = \tr(P^{a|x}\otimes Q^{b|y}\rho_{AB})
\end{equation}
for some shared state $\rho_{AB}$ and measurement effects $\{P^{a|x}\}$ and $\{Q^{b|y}\}$; and {\em nonsignaling} distributions (NS), whose marginals satisfy
\begin{subequations}\label{eq:ProbNS}
\begin{align}
    \sum_b p(a,b|x,y)& =  \sum_b p(a,b|x,y') \, ,\\
    \sum_a p(a,b|x,y)& =  \sum_a p(a,b|x',y) \, . 
\end{align}
\end{subequations}

\section{Properties of $\distquantum{\Delta}$}\label{app:ProofNQmonotone}
(1) Ref.~\cite[Corollary 4]{OurLong_2024} readily implies LOSR is a closed bounded subset of the euclidean space, and thus compact. Therefore the minimum in Eq.~\eqref{eq:ChanNLdef} is attained. 

(2) First assume that $\Delta$ is convex. This means, $\Delta(\sum_i \lambda_i \E_i, \sum_i \lambda_i \Phi_i) \leq \sum_i \lambda_i \Delta(\E_i, \Phi_i)$ for non-negative $\lambda_i$ adding up to one. In this case, we check that $\distquantum{\Delta}(\sum_i \lambda_i \E_i)  \leq \sum_i \lambda_i \distquantum{\Delta}(\E_i) $ thus $\distquantum{\Delta}$ is also convex. Indeed, for any collection $\Phi_i \in \operatorname{LOSR}$, also $\sum_i \lambda_i \Phi_i \in \operatorname{LOSR}$ by convexity of LOSR.

(3) Next, assume $\Delta$ is non-degenerate, namely $\Delta(\E, \Phi) = 0$ if and only if $\E = \Phi$. Then, $\distquantum{\Delta}(\E) = 0$ if and only if $\E \in \operatorname{LOSR}$ is faithful with respect to LOSR.

(4) Last, assume $\Delta$ is weakly contractive under superchannels, $\Delta(\Xi(\E), \Xi(\Phi)) \leq \Delta(\E, \Phi)$ for any superchannel $\Xi$. Assume $\Xi = \sum_\lambda p_\lambda\Xi_A^\lambda\otimes\Xi_B^\lambda$ is local. In this case, $\Xi(\Phi) \in \operatorname{LOSR}$ for each $\Phi \in \operatorname{LOSR}$. In particular, take $\Phi \in \operatorname{LOSR}$ such that $\Delta(\E, \Phi) = \distquantum{\Delta}(\E)$. Then,
\begin{equation}
\distquantum{\Delta}(\Xi(\E)) \leq \Delta(\Xi(\E), \Xi(\Phi)) \leq \Delta(\E, \Phi) = \distquantum{\Delta}(\E) .
\end{equation}
Thus, $\distquantum{\Delta}$ is monotone under local superchannels.

It is well known that the relative entropy \cite{schumacher2002relative} satisfies all these properties.
Moreover, any distance defined from a norm $\Delta(\E, \Phi) = \norm{\E - \Phi}$ satisfies the first three properties, also under the Choi-Jamio\l kowski isomorphism. Additionally, the fourth is also satisfied when the norm is contractive under linear operations. This is the case for the trace norm and the diamond norm \cite{Superchannels_Gour2019}. For completeness, we sketch the proof for the last.

\begin{claim}
The diamond norm is contractive under supperchannels.
\end{claim}
\begin{proof}
A superchannel $\Xi$ can be realized as $\Xi(\E) = \mathcal F \circ \E \otimes \id \circ \mathcal G$ for some channels $\mathcal F$ and $\mathcal G$ \cite[Theorem 1]{Superchannels_Gour2019}. The trace norm satisfies $\norm{\mathcal F (\bullet)}_1 \leq \norm{\mathcal F}_* \norm{\bullet}_1$, where $\norm{\mathcal F}_* = 1$ for channels. At the same time, $z = \mathcal G(x)$ satisfies $\norm{z}_1\leq \norm{\mathcal G}_* \norm{x}_1 \leq 1$. Then, 
\begin{align}
\norm{\Xi(\E)}_\diamond \leq \sup\{ \norm{\E\otimes \id (z)}_1 : \norm{z}_1 \leq 1 \} = \norm{\E}_\diamond \, .
\end{align}
\end{proof}

\begin{obs}
    For unitary channels, bounds for the quantity $\distquantum{1}$ can be obtained with semidefinite programming for unitary channels.
\end{obs}
\begin{proof}
Recall that the Choi matrix $J^U$ of a unitary channel has rank 1. Therefore, the Fidelity with respect to another Choi state $J^\Phi$ is a linear function, $F=\tr(J^UJ^\Phi)$, which in partciular can be optimized over $\Phi\in\text{LOSR}$ with the (strengthened) hierarchy of Appendix~\ref{app:hierarchyLOSR}. Being $d_{A_0B_0}$ the dimension of the input system, the Fuchs–van de Graaf inequality~\cite{FuchsGraffIneq_Fuchs1999} applied to Choi matrices establishes that $\|J^U-J^\Phi\|_1$ is lower bounded by $d_{A_0B_0}-\sqrt{\tr(J^U J^\Phi)}$ and upper bounded by $ \sqrt{d_{A_0B_0}-\tr(J^U J^\Phi)}$, whose optimal values can be approximated with semidefinite programming. This in turn allows to lower and upper bound $\distquantum{1}$, since the square root function is monotonic over its domain.
\end{proof}


\section{Proof of Proposition~\ref{prop:NcvsNq} and examples}\label{app:NcvsNq}
\begin{proof}
Consider a bipartite channel $\E$ with decoherent action $S^\E$ and a function $\Delta$ with corresponding decoherent measure $\delta$. Let $\tilde{\Phi}\in\text{LOSR}$ minimize $\Delta(\mathcal{D}\circ\E\circ\mathcal{D},\Phi)$ over $\Phi\in\text{LOSR}$. On the one hand, we have
\begin{align}
    \distquantum{\Delta}(\mathcal{D}\circ\E\circ\mathcal{D})&=\Delta(\mathcal{D}\circ\E\circ\mathcal{D},\tilde{\Phi})\nonumber\\
    &\leq \Delta(\mathcal{D}\circ\E\circ\mathcal{D},\mathcal{D}\circ\tilde{\Phi}\circ\mathcal{D})\,,\label{eq:appNq<Nc}
\end{align}
since $\mathcal{D}\circ\tilde{\Phi}\circ\mathcal{D}$ is an LOSR operation which is in general different from the minimizing LOSR operation $\tilde{\Phi}$. On the other hand, recall that by assumption $\Delta$ is contractive under superchannels $\Xi$, and in particular under concatenation with the dephasing $\mathcal{D}$. Therefore, we have the following inequality,
\begin{align}
     \distquantum{\Delta}(\mathcal{D}\circ\E\circ\mathcal{D})&\geq\Delta(\mathcal{D}^2\circ\E\circ\mathcal{D}^2,\mathcal{D}\circ\tilde{\Phi}\circ\mathcal{D})\nonumber\\
     &=\Delta(\mathcal{D}\circ\E\circ\mathcal{D},\mathcal{D}\circ\tilde{\Phi}\circ\mathcal{D})\,,\label{eq:appNq>Nc}
\end{align}
where the equality holds because $\mathcal{D}$ is a projection and therefore $\mathcal{D}^2:=\mathcal{D}\circ\mathcal{D}=\mathcal{D}$. Eqs.~\eqref{eq:appNq<Nc} and~\eqref{eq:appNq>Nc} together imply  
$\distquantum{\Delta}(\mathcal{D}\circ\E\circ\mathcal{D})=\Delta(\mathcal{D}\circ\E\circ\mathcal{D},\mathcal{D}\circ\tilde{\Phi}\circ\mathcal{D})$. This means that for the channel $\mathcal{D}\circ\E\circ\mathcal{D}$, there exists a decohered LOSR channel attaining the minimum in Eq.~\eqref{eq:ChanNLdef}.

Now, by definition of each  $\delta$ from decoherence of $\Delta$, we can write
\begin{equation}
\distclassical{\delta}(S^\E)=\min_{\Phi\in\text{LOSR}}\Delta(\mathcal{D}\circ\E\circ\mathcal{D},\mathcal{D}\circ\Phi\circ\mathcal{D}),
\end{equation}
since the whole local polytope can be obtained by decoherence of the set of LOSR operations~\cite{OurLong_2024}. 
Thus we arrive at
$\distquantum{\Delta}(\mathcal{D}\circ\E\circ\mathcal{D})=\distclassical{\delta}(S^\E)$, which shows item (i). To see item (ii) we note that
\begin{align}
    \distclassical{\delta}(S^\E)=\distquantum{\Delta}(\mathcal{D}\circ\E\circ\mathcal{D})\leq\distquantum{\Delta}(\E).
\end{align}
The equality holds due to item (i); the inequality holds because $\Delta$ is contractive under superchannels and therefore $\distquantum{\Delta}$ is monotonic under local superchannels (see Appendix~\ref{app:ProofNQmonotone}). In particular, $\mathcal{D}$ is a local superchannel.
\end{proof}

\begin{claim} Consider
$\E : \mathcal B(A_0 \otimes B_0) \to \mathcal B(A_1 \otimes B_1)$ a channel and $S:\ell_1(A_0 B_0) \to \ell_1(A_1 B_1)$ its decoherent action. Then, 
\begin{equation}
\label{eq:diamond_decoherence}
\norm{\mathcal D \circ \E \circ \mathcal D}_\diamond = \norm{S}_{11},
\end{equation}
where 
$
    \norm{\E}_\diamond  = \sup \{\norm{\E \otimes \id_{A_0B_0} (x)}_1 : \norm{x}_1 \leq 1\} 
$
and 
$
    \norm{S}_{11} = \sup \{\norm{S(a)}_{\ell_1} : \norm{a}_{\ell_1} \leq 1\} 
$
is the operator norm between sequence spaces $\ell_1$.
\end{claim}

\begin{proof}
First, given $a = \sum_i a_i \ket i \in \ell_1(A_0 \times B_0)$ with $\norm{a}_{\ell_1} \leq 1$ then $x = \sum_i a_i \ketbra{i}{i}\otimes \ketbra{0}{0}$ has $\norm{x}_1 \leq 1$. 
It is immediate to check that
\begin{equation}
\norm{\mathcal D\circ \E \circ \mathcal D \otimes \id_{XY}(x)}_1 = \norm{S(a)}_{\ell_1}  \, ,
\end{equation}
thus $\norm{\mathcal D \circ \E \circ \mathcal D}_\diamond \geq \norm{S}_{11}$.

Conversely, given $x = \sum_{ij} x_{ij} \ketbra{i}{j} \otimes A_{ij}$ with $\norm{x}_1 \leq 1$ we have that $\mathcal D(x) = \sum_{i} x_{ii} \ketbra{i}{i} \otimes A_{ii}$ satisfies $\norm{\mathcal D(x)}_1  = \sum_i |x_{ii}| \norm{A_{ii}}_1 \leq 1$. This is a consequence of Gelfand-Naimark theorem \cite[Theorem III.4.5]{bhatia2013matrix}, since $\mathcal D$ is a projection. Then $a = \sum_i x_{ii} \norm{A_{ii}}_1 \ket i $ satisfies $\norm{a}_{\ell_1} \leq 1$. Furthermore, we can check
\begin{equation}
\norm{\mathcal D\circ \E \circ \mathcal D \otimes \id_{XY}(x)}_1 \leq \norm{S(a)}_{\ell_1}  \, ,
\end{equation}
thus, $\norm{\mathcal D \circ \E \circ \mathcal D}_\diamond \leq \norm{S}_{11}$. Indeed,
\begin{equation}
\sum_j \norm{\sum_i x_{ii} S_{ji} A_{ii}}_1
\leq \sum_j \left| \sum_i x_{ii} S_{ji} \norm{A_{ii}}_1 \right| \, .
\end{equation}

\end{proof}

Notice that from Eq.~\eqref{eq:diamond_decoherence} and the fact that LOSR is mapped to L under decoherence it follows that $\distquantum{\diamond}(\mathcal D\circ \E \circ \mathcal D) = \distclassical{\diamond}(S^\E)$ as in Proposition~\ref{prop:NcvsNq}. 
The same holds for the operator norm between spaces with the trace norm, $\norm{\mathcal D \circ \E \circ \mathcal D}_{11} = \norm{S}_{11}$. That is, the dilation in the diamond norm makes no difference under decoherence.

\begin{claim}
The trace norm of the Choi state reduces under decoherence to the element-wise one norm of the decoherent action,
\begin{equation}
\label{eq:trace_decoherence}
    \norm{J(\mathcal D \circ \E \circ \mathcal D)}_1 = \norm{S^\E}_{\ell_1} \, ,
\end{equation}
where $\norm{S}_{\ell_1} = \sum_{ij}|S_{ij}|$.
\end{claim}

\begin{proof}
This follows from the fact that the decoherent action $S^\E$ coincides with the diagonal of $J(\mathcal D \circ \E \circ \mathcal D)$ and the off-diagonal terms are zero.
\end{proof}

Again, from Eq.~\eqref{eq:trace_decoherence} and the fact that LOSR is mapped to L under decoherence it follows that $\distquantum{1}(J(\mathcal D\circ \E \circ \mathcal D)) = \distclassical{1}(S^\E)$ as in Proposition~\ref{prop:NcvsNq}.

\begin{claim}
The relative entropy of the Choi state reduces under decoherence to the classical relative entropy of the decoherent action,
\begin{equation}
\label{eq:entropy_decoherence}
    H(J(\mathcal D \circ \E \circ \mathcal D), J(\mathcal D \circ \Phi \circ \mathcal D)) = h(S^\E | S^\Phi) \, ,
\end{equation}
where $h(S| T) = \sum_{ij} S_{ij} (\log S_{ij} - \log T_{ij})$.
\end{claim}
\begin{proof}
This follows from the fact that the decoherent action $S^\E$ coincides with the diagonal of $J(\mathcal D \circ \E \circ \mathcal D)$ and the off-diagonal terms are zero.
\end{proof}

Again, from Eq.~\eqref{eq:entropy_decoherence} and the fact that LOSR is mapped to L under decoherence it follows that $\distquantum{H}(J(\mathcal D\circ \E \circ \mathcal D)) = \distclassical{H}(S^\E)$ as in Proposition~\ref{prop:NcvsNq}.

\section{A hierarchy approximating LOSR}\label{app:hierarchyLOSR}
With the Choi-Jamio\l kowski isomorphism~\cite{choi1975completely,jamiolkowski1972linear}, deciding whether or not a given bipartite channel $\E_{AB}$ is LOSR is a constrained version of the separability problem: one asks that the Choi matrix defines a {\em separable operation} (SEP),
\begin{equation}\label{eq:appSepChoi}
    J_{A_0B_0A_1B_1}=\sum_\lambda p_\lambda J^\lambda_{A_0A_1}\otimes J^\lambda_{B_0B_1}\,,
\end{equation}
where $p_\lambda\geq 0$ and $J^\lambda_{A_0A_1},J^\lambda_{B_0B_1}\geq 0$~\cite{doherty2004complete, Navascues_PPT_DPS_2009}, and the additional linear constraints $\tr_{A_1}(J_{A_0A_1}^\lambda)=\id_{A_0}$ and $\tr_{B_1}(J_{B_0B_1}^\lambda)=\id_{B_0}$ to ensure that each factor defines a completely positive and trace preserving (CPTP) map. A solution~\cite{ChanFidSymm_Chee2024} can be found by applying the hierarchy of~\cite{berta2021semidefinite}, which consists of strengthening the standard Doherty-Parrilo-Spedalieri (DPS) hierarchy~\cite{doherty2004complete,Navascues_PPT_DPS_2009} with further linear constraints at each level $n$: 
\begin{align}
    \text{find} & \quad X\geq 0\label{eq:appCP} \\
    &\quad \tr_{A_1B_1^{0}B^1...B^{n}}(X)=\id_{A_0B_0}\label{eq:appTP}\\
    &\quad V_{B^{i},B^{n}}XV_{B^{i},B^{n}}^\dag = X \quad\forall\,i\label{eq:appSYM}\\
    &\quad X^{T_{B^{n}}}\geq 0\label{eq:appPT}\\
    &\quad \tr_{B_1^{n}}(X)=\tr_{B_0^{n}B_1^{n}}X\otimes\id_{B_0^{n}}\label{eq:appCONS1}
    \\
    &\quad \tr_{A_0}(X)=\id_A\otimes\tr_{A_0A_1}(X)\label{eq:appCONS2}
\end{align}
where $X$ 
acts on an extended Hilbert space $\mathcal{H}_{A_0B_0A_1B_1}\otimes(\mathcal{H}_{B_0B_1})^{\otimes n-1}$, we denote $A=A_0A_1$ and $B^{(0)}=B_0B_1$, $X_{A_0B_0A_1B_1}=\tr_{B^1...B^n}(X)$, $V_{B^{(i)},B^{(n)}}$ is the SWAP operator between the $i$'th and $n$'th party and $T_{B^{(n)}}$ is the partial transpose on the party $B^{(n)}$. Here Eqs.~\eqref{eq:appCP} and~\eqref{eq:appTP} impose that the marginal $X_{A_0B_0A_1B_1}$ is the Choi matrix of a CPTP map; Eqs.~\eqref{eq:appSYM} and~\eqref{eq:appPT} converge to the condition that $X_{A_0B_0A_1B_1}$ decomposes as~\eqref{eq:appSepChoi}~\cite{doherty2004complete,Navascues_PPT_DPS_2009}, which means that $X_{A_0B_0A_1B_1}$ defines a separable operation, $X_{A_0B_0A_1B_1}=\sum_\lambda X_{A_0A_1}^\lambda\otimes X_{B_0B_1}^\lambda$; Eqs.~\eqref{eq:appCONS1} and~\eqref{eq:appCONS2} impose that $X_{A_0B_0A_1B_1}$ is the Choi matrix of a qNS channel, and converge at some $n$ to the condition that the marginals of each term satisfy $X^\lambda_{A_0}=\id$ and $X^\lambda_{B_0}=\id$~\cite{berta2021semidefinite}, which means that they define trace preserving channels and therefore $X_{A_0B_0A_1B_1}$ defines an LOSR channel. 

Consider a diagonal Choi matrix given by $J^\E=\sum_{xyab}p(ab|xy)\dyad{xyab}{xyab}$ with $p=S^\E\in\text{NS}$
, e.g. given by the Popescu-Rohrlich (PR-) box 
\begin{equation}\label{eq:CAprBox}
S^\E=\frac{1}{2}
\begin{pmatrix}
    1 & 1 & 1 & 0 \\
    0 & 0 & 0 & 1 \\
    0 & 0 & 0 & 1 \\
    1 & 1 & 1 & 0
\end{pmatrix}\,,
\end{equation}
Such Choi matrix is in the set SEP because it is diagonal in product basis, and thus is not detected by the Doherty-Parrilo-Spedalieri (DPS) hierarchy (Eqs.~\eqref{eq:appCP}-\eqref{eq:appPT}). At the first level, equations~\eqref{eq:appCONS1} and~\eqref{eq:appCONS2} are also satisfied because a classical map given by a stochastic matrix in the nonsignaling polytope defines a quantum nonsignaling channel. However, such an example does not satisfy Eq.~\eqref{eq:appLochier}.

\section{Proof of Proposition~\ref{prop:SimulStoch}}\label{app:SimulStoch}
\begin{proof}
Consider an LOSE operation as in Eq.~\eqref{eq:Localizable}. 
Let the Kraus representation of the channels $\E_{AR}$ and $\E_{BS}$, which act through $\E_{AR}(\cdot)=\sum_{i,j}K^{\E_{AR}}_{i,j}\cdot{K^{\E_{AR}}_{i,j}}^\dag$ and $\E_{BS}(\cdot)=\sum_{k,l}Q^{\E_{BS}}_{k,l}\cdot{Q^{\E_{BS}}_{k,l}}^\dag$, be $ K^{\E_{AR}}_{i,j}=\dyad{i}{j}_A\otimes{L^{i|j}}_R$ and $Q^{\E_{BS}}_{k,l}=\dyad{k}{l}_B\otimes{T^{k|l}}_S$, 
where $M^{i|j}:={L^{i|j}}^\dag L^{i|j}$ and $N^{k|l}:={T^{k|l}}^\dag T^{k|l}$ are conditional measurement operators. One verifies that the decoherent action $S^{\E}$ of $\E_{AB}$ reads entrywise $S^{\E}_{ab,xy}=\tr(M^{a|x}\otimes N^{b|y}\tau_{RS})$, obtaining a similar result as if Alice and Bob were able to measure in coherent bases to implement a standard Bell game~\cite{BrunnerReviewBellNonloc_2014}. 
Therefore, Alice and Bob can simulate a classical channel beyond the set of matrices given by Eq.~\eqref{eq:losr-stochastic-matrix} without communication if they have access to quantum resources.
\end{proof}

\bibliography{Bibliography}

\end{document}